\documentclass[conference,twocolumn]{IEEEtran}
\IEEEoverridecommandlockouts

\usepackage{fullpage,graphicx,psfrag,amsmath,amsfonts,verbatim}
\usepackage[small,bf]{caption}
\usepackage{amsthm}
\usepackage{mathtools,amssymb,color}

\newcommand{\reals}{\mathbb{R}}

\newcommand{\argmin}{\mathop{\rm argmin}}

\newcommand{\bX}{\mathbf{X}}

\newcommand{\bp}{\mathbf{p}}
\newcommand{\bx}{\mathbf{x}}
\newcommand{\bbP}{\mathbb{P}}

\newcommand{\bz}{\mathbf{z}}
\newcommand{\bQ}{\mathbf{Q}}

\allowdisplaybreaks

\ifCLASSINFOpdf
\else
\fi
\newtheorem{thm}{Theorem}
\newtheorem{rmk}{Remark}

\begin{document}
%
\title{ Minimum HGR Correlation Principle: \\From Marginals to Joint Distribution}

\author{\IEEEauthorblockN{Farzan Farnia}
\IEEEauthorblockA{farnia@stanford.edu}
\thanks{This work is partially supported by a Stanford Graduate Fellowship, and the Center for Science of Information (CSoI), an NSF
Science and Technology Center, under grant agreement CCF-0939370.}
\and
\IEEEauthorblockN{Meisam Razaviyayn}
\IEEEauthorblockA{meisamr@stanford.edu}
\and
\IEEEauthorblockN{Sreeram Kannan}
\IEEEauthorblockA{ksreeram@ee.washington.edu}
\and
\IEEEauthorblockN{David Tse}
\IEEEauthorblockA{dntse@stanford.edu}}


%


\linespread{0.95}

\maketitle

\begin{abstract}
 Given  low order moment information over the random variables $\bX = (X_1,X_2,\ldots,X_p)$ and $Y$, what distribution  minimizes the Hirschfeld-Gebelein-R\'{e}nyi (HGR) maximal correlation coefficient between $\bX$ and $Y$, while remains faithful to the given moments?  The answer to this question is important especially in order to fit models over $(\bX,Y)$ with minimum dependence among the random variables $\bX$ and $Y$. In this paper, we investigate this question first in the continuous setting by showing that the jointly Gaussian distribution achieves the minimum HGR correlation coefficient among distributions with the given first and second order moments. Then, we pose a similar question in the discrete scenario by fixing the pairwise marginals of the random variables $\bX$ and $Y$. Subsequently, we derive a lower bound for the HGR correlation coefficient over the class of distributions with fixed  pairwise marginals. Then we show that this lower bound is tight if there exists a distribution with certain {\it additive } structure satisfying the given pairwise marginals. Moreover, the distribution with the additive structure achieves the minimum HGR correlation coefficient. Finally, we conclude by showing that the event of obtaining pairwise marginals containing an additive  structured distribution has a positive Lebesgue measure over the probability simplex.
\end{abstract}
\IEEEpeerreviewmaketitle

\section{Introduction}
A well-known measure of dependence between two random variables $X$ and $Y$ is the  Pearson correlation coefficient which is defined as
\[
\rho (X,Y) = \frac{{\rm Cov}(X,Y)}{\sqrt{{\rm Var}(X) {\rm Var}(Y)}},
\]
assuming that $0<{\rm Var}(X), {\rm Var}(Y) < \infty$. Clearly, this measure of dependence is zero if $X$ and $Y$ are independent; while the converse is not necessarily true in general. Furthermore, this measure of dependence fails in discovering nonlinear dependence among random variables.  A closely related measure of dependence, which was first introduced by Hirschfeld and Gebelein \cite{hirschfeld1935connection,gebelein1941statistische} and then studied by R\'{e}nyi \cite{renyi1959measures}, is the HGR maximal  correlation coefficient defined as: 
\begin{equation}
\label{eq:RenyiDef}
{\rho_{m}}(X,Y) \triangleq \sup_{f,g } \mathbb{E} \left[ f(X) g(Y)\right],
\end{equation}
where the maximization is taken over  the class of all measurable functions $f$ and $g$ with the property that $\mathbb{E} [f(X)] = \mathbb{E} [g(Y)] = 0$ and  $\mathbb{E} [f^2(X)] = \mathbb{E}[g^2(Y)] = 1$. 
The HGR  correlation coefficient has many natural properties one would look for as a measure of dependence. For example, the  correlation coefficient of two random variables $X$ and $Y$ is  normalized to be between  $0$ and $1$. Furthermore, this coefficient is zero if and only if the two random variables are independent; and it is one if there is a strict dependence between $X$ and $Y$; see \cite{renyi1959measures} for other interesting properties of the HGR correlation. 

In the classical prediction problems the task is to predict the value of a random variable $Y$ based on the observations on the random variable $X$. When the probability model on the random variables $X$ and $Y$ is not known,  the first step in prediction is to fit a probabilistic model on the random variables based on the knowledge obtained from training data or other sources. A popular approach to fit a model for inference/prediction is to use the maximum entropy principle \cite{jaynes1957information}. This principle states that given a set of constraints on the ground truth distribution, the distribution with the maximum (Shannon) entropy under those constraints is a {\it ``proper"} representer of the class. In practice, this idea can be implemented by estimating low order marginals from data and find a distribution with maximum entropy satisfying the low order marginals. The maximum entropy principle is in essence the sprit of {\it variational method} in  graphical models. Applying similar idea to the prediction problem, one approach is to find a distribution that maximizes the conditional entropy of the random variable $Y$ given  $X$ under a fixed set of marginals. Given the marginal distribution of $Y$,  this principle is equivalent to minimizing the mutual information between the two random variables $X$ and $Y$; see, e.g., \cite{globerson2004minimum}, for which no efficient computational approach is known in the literature for high dimensional problems. When HGR  correlation is used instead of mutual information, the problem of interest is to find the distribution over $(X,Y)$ with the minimum HGR correlation between the two random variables $X$ and $Y$ which stays faithful to the estimated low order marginals.  

To answer the question of finding the distribution with minimum HGR correlation coefficient, we start by the trivial inequality ${\rho_{m}}(X,Y) \geq |\rho (X,Y)|$ which holds for any two random variables $X$ and $Y$. Interestingly,   this inequality is tight when $X$ and $Y$ are jointly Gaussian \cite{gebelein1941statistische,lancaster1957some, renyi1959measures}. In other words,  among all possible distributions with fixed first and second order moments on $X$ and $Y$, the jointly Gaussian distribution on $(X,Y) $ achieves the minimum HGR correlation.   Similar question rises in the case where the observed random variable in prediction is in the vector form:   Let $\bX = (X_1,X_2, \ldots,X_p)$ be the observed random vector and $Y$ be the random variable denoting the prediction target. Since in many recent  prediction problems in statistics and machine learning, the random vector $\bX$ is in a  large dimensional space, estimating the complete joint distribution of $\bX$ and $Y$ is not computationally and statistically feasible from the data. Motivated by the variational approach in  graphical models,  one can estimate the first and second order moments of $(\bX,Y)$ and fit a model which minimizes the dependence between $\bX$ and $Y$ while stay faithful to the constraints on the first and second order moments. Therefore, the question of interests in this paper is as follows:

\noindent{\bf Question:} Given the first and second order moments of the two random variables $\bX=(X_1,\ldots,X_p)$ and $Y$, what is the distribution with minimum HGR  correlation coefficient between $\bX$ and $Y$?  

In this work, we investigate this question in both continuous and discrete scenarios. In particular, in section~\ref{sec:continuous} we answer  this question for the continuous case where both of the random variables $\bX$ and $Y$ are continuous. Then, we cast a similar question in the discrete setting and partially answer it by first finding a lower bound for the minimum HGR correlation coefficient. Then, we show that, under a certain {\it additive structure condition}, this lower bound is tight and  can be achieved through a certain probability distribution satisfying the pairwise constraints.

\section{Minimum HGR correlation distribution: continuous setting}
\label{sec:continuous}
Let us assume that the first order moment $\mathbb{E} [(\bX \; \;Y)] = \boldsymbol{\mu} \in \reals^{p+1}$ and the second order moment $\mathbb{E} [(\bX \;\;Y)^T  (\bX\;\; Y)] = \boldsymbol{\Lambda} \in \reals^{(p+1) \times (p+1)}$  are given. 
Our goal is to find the probability distribution $\bbP^*_{\bX,Y}(\bx,y)$ on the random variables $\bX$ and $Y$ which is the solution to the following optimization problem:
%
%
%
%
\begin{equation}
\label{eq:RenyiMinCont}
\begin{split}
\min_{\mathbb{P}_{\bX,Y}} \quad &{\rho_{m}} (\bX,Y) \\
{\rm s.t.}\quad &\mathbb{E}_{\bbP}[(\bX \;\;Y)] =  \boldsymbol{\mu}\\
&\mathbb{E}_\bbP [(\bX \;\;Y)^T  (\bX\;\; Y)] = \boldsymbol{\Lambda}.
\end{split}
\end{equation}
The following simple result shows that the joint Gaussian distribution is always a solution of \eqref{eq:RenyiMinCont}. The result and its proof is straightforward, but since we could not find it in the literature, we state it here. This result  plus the additive property of the Gaussian distribution, which will be explained later,  shed light on the discrete scenario as well. 

%
%
\begin{thm}
Let $\bbP^G_{\bX,Y}$ be the Gaussian probability distribution defined over $\bX$ and $Y$ with the given first and second order moments $\boldsymbol{\mu}$ and $\boldsymbol{\Lambda}$. Then $\bbP^G_{\bX,Y}$ is a minimizer of \eqref{eq:RenyiMinCont}, i.e., 
\begin{align}
\bbP^G_{\bX,Y} \in \argmin_{\bbP_{\bX,Y} } \quad & {\rho_{m}} (\bX,Y) \nonumber\\
{\rm s.t.} \quad&\mathbb{E}_{\bbP}[(\bX \;\;Y)] =  \boldsymbol{\mu} \nonumber\\
&\mathbb{E}_\bbP [(\bX \;\;Y)^T  (\bX\;\; Y)] = \boldsymbol{\Lambda} \nonumber
\end{align}
\end{thm}
\begin{proof}
Under the given first and second order moments, the random variable $Y$ can be written as
\begin{equation} \label{SumFormulation}
Y-{{\mu}}_{Y}={\mathbf{a}}^T (\mathbf{X}-{{\boldsymbol{\mu}}}_{\mathbf{X}})+ Z,
\end{equation}
for some vector $\mathbf{a}\in \reals^p$ with the random variable $Z$ being uncorrelated with  $\mathbf{X}$.  Here $\mu_Y$ and $\boldsymbol{\mu}_\mathbf{X}$ are the given first order moments of $Y$ and $\bX$, respectively. Notice that the  vector $\mathbf{a}$ and thus $\text{\rm Var}\left({{\mathbf{a}}^T \mathbf{X}}\right)$ are completely determined by the first and second order marginals $\boldsymbol{\mu}$ and $\boldsymbol{\Lambda}$. Therefore, taking the functions $f$ and $g$ to be
\begin{equation}
f(\mathbf{X})= \frac{ {\mathbf{a}}^T (\mathbf{X}-\boldsymbol{\mu}_\bX) } { \sqrt{{\rm Var}\left({{\mathbf{a}}^T \mathbf{X}}\right)}},\quad 
g(Y)=\frac{Y- \mu_Y}{\sqrt{{\rm Var}(Y)}}
\end{equation}
in the HGR correlation definition \eqref{eq:RenyiDef} leads to the following lower bound for the minimum value of \eqref{eq:RenyiMinCont}:
\begin{equation}
\label{eq:LowBdTemp}
\mathbb{E}\left(f(\mathbf{X})g(Y)\right)=\sqrt{\frac{{\rm Var}\left({{\mathbf{a}}^T \mathbf{X}}\right)}{{\rm Var}(Y)}}.
\end{equation}
Hence the value of the obtained lower bound \eqref{eq:LowBdTemp} only depends on the first and second order moments of $(\bX,Y)$. Now we show that the jointly Gaussian distribution with  the given first and second order moments achieves the above lower bound. Under the jointly Gaussian distribution, the random variable $Z$ in \eqref{SumFormulation} becomes independent from  $\bX$ and therefore
\begin{equation}
\nonumber
(X_1,\ldots , X_p) \, \stackrel{\rule{0.3cm}{0.5pt}}{} \, {{\mathbf{a}}^T \mathbf{X}} \, \stackrel{\rule{0.3cm}{0.5pt}}{} \, Y
\end{equation}
forms a Markov chain.  Hence,  using the alternative conditional definition of the HGR  correlation \cite{anantharam2013maximal}, we obtain
\begin{equation}
\nonumber
\begin{split}
&\;{\rho_{m}}({\mathbf{X};Y})  = \max \; \;{\mathbb{E}\left( \mathbb{E}^2(g(Y)\vert \mathbf{X})\right)} \\
& = \max \;\;{\mathbb{E}\left( \mathbb{E}^2(g(Y)\vert \mathbf{X},\mathbf{a}^T\mathbf{X})\right)} \\
& = \max\;\; {\mathbb{E}\left( \mathbb{E}^2(g(Y)\vert \mathbf{a}^T\mathbf{X})\right)}  = {\rho_{m}}\left({\mathbf{a}^T\mathbf{X};Y}\right),
\end{split}
\end{equation}
where the maximizations are taken over the functions $g(\cdot)$ satisfying $\mathbb{E}(g(Y)) =0\, \text{\rm and}\, \mathbb{E}(g^2(Y)) =1$. 
Clearly, $(\mathbf{a}^T\mathbf{X},Y)$ is distributed according to a jointly Gaussian distribution. Thus,
$
{\rho_{m}}({\mathbf{X};Y}) = {\rho_{m}}\left({\mathbf{a}^T\mathbf{X};Y}\right) = \sqrt{\frac{{\rm Var}\left({{\mathbf{a}}^T \mathbf{X}}\right)}{{\rm Var}(Y)}},
$
which completes the proof.
\end{proof}
Notice that the Gaussian distribution has the property that $\mathbb{E} \left[ Y \big| \bX \right] = \sum_{i=1}^pf_i(X_i)$ for some linear functions $f_i$. As will be seen in the next section, this additive model property \cite{friedman1981projection,buja1989linear,hastie2009elements} plays an important role in the  discrete scenario as well. 

%
%
%
\section{Minimum HGR correlation distribution: discrete setting}
Let us consider the scenario where the random variables $\bX$ and $Y$ are discrete. Here, motivated by the standard binary classification problem in machine learning, we assume that the random variable $\bX  = (X_1,X_2,\ldots,X_p)\in \mathbb{X}^p$ is in a categorical structure  with $|\mathbb{X}|<\infty$, while the random variable $Y \in \mathbb{Y} \triangleq \{0,1\}$ is binary. 
Since  the random variables are categorical, the relabeling of the alphabets $\mathbb{X}$ and $\mathbb{Y}$ should not affect any model fitting approach. Noticing that the first and second order moments of  $(\bX,Y)$  are not invariant to relabeling of the alphabets, we   consider the class of distributions with a fixed given pairwise marginal distribution of $(X_1,\ldots,X_p,Y)$ instead of fixing the moments.  This modification is equivalent to fixing the first and second order moments of the indicator/dummy variables  \cite{suits1957use} defined over our categorical data. More precisely, defining the random vector $\widetilde{\bX} = (\widetilde{X}_i^x)_{i,x} $ with $\widetilde{X}_i^x = 1$ if $X_i = x$ and  $\widetilde{X}_i^x = 0$ otherwise,  the first and second order moment knowledge on $\widetilde{\bX}$ is equivalent to  the pairwise marginals knowledge on  $X_1,\ldots,X_p$. 
Therefore, finding the distribution with the minimum HGR correlation coefficient between $\bX$ and $Y$ can be formally stated as 
\begin{equation}
\label{eq:MinRenyiPairwise}
\begin{split}
\min_{\mathbb{P}_{\bX,Y}} \quad {\rho_{m}} (\bX,Y) \quad \quad {\rm s.t.}\quad \bbP_{\bX,Y} \in \mathcal{C},
\end{split}
\end{equation}
where $\mathcal{C}$ is the class of distributions with given pairwise marginals defined as
\begin{align}
& \mathcal{C} = \bigg\lbrace \bbP_{\mathbf{X},Y}\, \bigg|  \,\bbP(X_i = x_i, X_j = x_j) = \mu_{x_i x_j}^{ij}, \nonumber\\
& \bbP(X_i = x_i, Y = y) = \mu_{x_i y}^{i}, \;\forall x_i,x_j \in \mathbb{X},\;\forall y \in \mathbb{Y}, \, \forall i,j \bigg\rbrace. \nonumber
\end{align}

The optimization problem \eqref{eq:MinRenyiPairwise} is convex in terms of the joint distribution $\mathbb{P}_{\mathbf{X},Y}$; however, the number of variables grows exponentially in $p$. To deal with this exponential computational complexity and 
finding the solution of \eqref{eq:MinRenyiPairwise} indirectly, let us first define a lower bound on the HGR correlation coefficient between the two random variables $\bX$ and $Y$ with a given joint probability distribution $\bbP_{\bX,Y}$:  
\begin{equation}
\label{eq:RlbDef}
\begin{split}
\rho_{m}^{\rm lb}(\bX,Y) \triangleq \max_{f,g}  \;\; &\mathbb{E} [f(\bX) g(Y)] \\
{\rm s.t. }\quad   & f \in \mathcal{F}, \\
& \mathbb{E}_\bbP [f(\bX)] = \mathbb{E}_\bbP [g(Y)] = 0, \\
  &\mathbb{E}_\bbP [f^2(\bX)] = \mathbb{E}_\bbP[g^2(Y)] = 1, 
  \end{split}
\end{equation}
where $\mathcal{F}$ is the class of separable functions defined as
\[
\mathcal{F} \triangleq  \left\{f \,\big|\, f(\bX) = \sum_{i=1}^p \xi_i(X_i)\;{\rm with}\; \xi_i:\mathbb{X}\mapsto \reals\right\}.
\]
Clearly, $\rho_{m}^{\rm lb} (\bX,Y) \leq {\rho_{m}}(\bX,Y)$ since $\rho_{m}^{\rm lb} (\bX,Y)$  is obtained by restricting the feasible set in \eqref{eq:RenyiDef}.  The following theorem shows that the value of $\rho_{m}^{\rm lb} (\bX,Y)$ is efficiently computable.

\begin{thm}
\label{thm:RenyiLB}
Suppose $\mathcal{C} \neq \emptyset $ and let us without loss of generality assume that $ \mathbb{X} = \{1,2,\ldots,m\}$. Then
\begin{equation} \label{WCRC_new: EquiQuad}
\rho_{m}^{\rm lb}(\bX,Y) =  \sqrt{1 - \frac{\gamma^{\rm lb}}{\bbP(Y=0) \bbP(Y=1)}},
\end{equation}
where  
\begin{align}
\gamma^{\text{\rm lb}} \triangleq \min_{\bz \in \reals^{pm\times 1}}\;\; \bz^T \bQ \bz + \mathbf{d}^T \bz + \frac{1}{4}, \label{eq:QuadraticProblem}
\end{align}
with
$\bQ \in \reals^{pm\times pm}$ and  $\mathbf{d} \in \reals^{pm\times 1}$  are defined as 
\begin{align}
&\bQ_{m(i-1)+k,m(j-1)+\ell} =\bbP(X_i = k,X_j = \ell), \nonumber \\
&\mathbf{d}_{m(i-1)+k} =\bbP(X_i = k,Y=1) - \bbP(X_i=k, Y = 0), \nonumber
\end{align}
for every $i,j = 1,\ldots,p$ and $k,\ell = 1,\ldots,m$.
\end{thm}
\begin{proof}
The proof is provided in the Appendix. 
\end{proof}
\begin{rmk}\label{Remark: QuadraticEquivalenceIndicator}
Assume the pairwise marginals are estimated from a given dataset containing $n$ data points.  Let $\mathbf{w}^j \in \{0,1\}^{mp\times 1}$ be the indicator of the $j$-th datapoint with $(\mathbf{w}^j)_{m(i-1)+ k} = 1$ if $X_i = k$ in the $j$-th datapoint and $(\mathbf{w}^j)_{m(i-1)+ k} = 0$ otherwise.  Define the vector $\mathbf{b} \in \{-\frac{1}{2}, \frac{1}{2}\}^{n \times 1}$ with $\mathbf{b}_j = \frac{1}{2}$ if the random variable $Y = 1$ in the $j$-th datapoint and $\mathbf{b}_j = -\frac{1}{2}$, otherwise. Then the  optimization problem~\eqref{eq:QuadraticProblem}, is equivalent to the following least squares regression problem
\begin{align}
\min_{\bz}\quad \|\mathbf{W}\bz - \mathbf{b}\|_2^2, \nonumber
\end{align}
where $\mathbf{W} = [\mathbf{w}^1  \mathbf{w}^2 \ldots \mathbf{w}^n]^T$.
\end{rmk}
Theorem~$\ref{thm:RenyiLB}$ simply states that the lower bound $\rho_{m}^{\rm lb}(\bX,Y)$ is easily computable by solving a convex optimization problem; see \eqref{WCRC_new: EquiQuad}, \eqref{eq:QuadraticProblem}. Moreover, this lower bound only depends on the pairwise marginals of the distribution of $\bbP_{\bX,Y}$. Consequently, $\rho_{m}^{\rm lb}(\bX,Y)$ is the same across all distributions in $\mathcal{C}$ and therefore it is well-defined to denote $\rho_{m}^{{\rm lb},\mathcal{C}}$ as the lower bound $\rho_{m}^{\rm lb}(\bX,Y)$ achieved by any of the  distributions in $\mathcal{C}$. Furthermore, this quantity is also a lower bound for the optimum value of \eqref{eq:MinRenyiPairwise}, i.e., 
$
\rho_{m}^{{\rm lb},\mathcal{C}} \leq \displaystyle{\min_{\mathbb{P}_{\bX,Y}\in \mathcal{C}} }\quad {\rho_{m}} (\bX,Y).
$

The following theorem provides an interpretative necessary and sufficient condition under which this lower bound is tight. Subsequently, in Theorem \ref{thm:SuffNecCond}, we provide a computationally efficient  approach to verify this condition.

\begin{thm}
\label{thm:SuffNecCond_part1}
The achieved lower-bound $\rho_{m}^{{\rm lb},\mathcal{C}} $  is tight, i.e.,
\[
\rho_{m}^{{\rm lb},\mathcal{C}} = \min_{\bbP_{\bX,Y}\in \mathcal{C}} {\rho_{m}}(\bX,Y),
\]
 if and only if there exists a probability distribution $\bbP \in \mathcal{C}$ for which the conditional expectation $\mathbb{E}_\bbP \left[Y\, \big| \, \bX\right]$ takes a separable form, i.e. $f_i$'s exist such that
\begin{equation}
\mathbb{E}_\bbP \left[Y\, \big| \, \bX\right] = \sum_{i=1}^p f_i(X_i).
\end{equation}
\end{thm} 
\begin{proof}
First, consider the probability distribution $\bbP$ for which $\mathbb{E}_\bbP \left[Y\, \big| \, \bX\right] = \sum_{i=1}^p f_i(X_i)$. Notice that, when $Y$ is binary, the function $g(Y) = \frac{Y - \mathbb{E}[Y]}{\sqrt{{\rm Var}(Y)}}$ is the only  feasible function $g(\cdot)$ in \eqref{eq:RenyiDef}. Therefore, the HGR correlation coefficient between $\bX$ and $Y$ can be calculated by
\begin{align}
{\rho_{m}}(\bX,Y) =  \max_{f} \quad & \mathbb{E} \left[f(\bX) \frac{Y - \mathbb{E}[Y]}{\sqrt{{\rm Var}(Y)}} \right]  \label{eq:BinRenyi}\\
{\rm s.t.}\;\; & \mathbb{E} \left[f(\bX)\right] = 0, \;\;\mathbb{E} \left[f^2(\bX)\right] = 1, \nonumber
\end{align}
where the expectations are taken with respect to the probability distribution $\bbP$. Furthermore, the objective can be rewritten as
\begin{align}
& \mathbb{E} \left[f(\bX) \frac{Y - \mathbb{E}[Y]}{\sqrt{{\rm Var}(Y)}} \right] =   \mathbb{E} \left[\mathbb{E}\left[f(\bX) \frac{Y - \mathbb{E}[Y]}{\sqrt{{\rm Var}(Y)}} \,\big| \, \bX\right]\right]  \nonumber\\
=\; & \mathbb{E} \left[f(\bX) \mathbb{E}\left[\frac{Y - \mathbb{E}[Y]}{\sqrt{{\rm Var}(Y)}} \,\big| \, \bX\right]\right] \\
= \; &\mathbb{E} \left[ f(\bX) \left(\frac{\sum_{i=1}^p f_i(X_i) - \mathbb{E}[\sum_{i=1}^pf_i(X_i)]}{\sqrt{{\rm Var}(Y)}}\right)\right].\nonumber
\end{align}
A simple application of the Cauchy-Schwarz inequality implies that the optimizer of  \eqref{eq:BinRenyi} is of the form
\[
f^*(\bX) = \frac{\sum_{i=1}^p f_i(X_i) - \mathbb{E}[\sum_{i=1}^p f_i(X_i)] }{\sqrt{{\rm Var} (\sum_{i=1}^p f_i(X_i))}},
\]
which is in a separable form. Therefore, $f^*(\cdot)$ is feasible to \eqref{eq:RlbDef} and ${\rho_{m}}(\bX,Y)  = \rho_{m}^{\rm lb}(\bX,Y)$.

To show the other direction, notice that the above Cauchy-Schwarz inequality holds with equality if and only if for a constant $c$
\begin{equation}
f(\bX) = c\,\mathbb{E}\left[\frac{Y - \mathbb{E}[Y]}{\sqrt{{\rm Var}(Y)}} \,\big| \, \bX\right],
\end{equation}
with probability one. Hence if ${\rho_{m}}(\bX,Y)  = \rho_{m}^{\rm lb}(\bX,Y)$ there is a separable function function $f^*(\bX)$, a probability distribution $\mathbb{P}^* \in \mathcal{C}$, and a constant $c^*$ such that
\begin{equation}
f^*(\bX) = c^*\,\mathbb{E}_{\mathbb{P}^*}\left[\frac{Y - \mathbb{E}[Y]}{\sqrt{{\rm Var}(Y)}} \,\big| \, \bX\right]
\end{equation}
which implies that $\mathbb{E}_{\mathbb{P}^*}\left[Y\,\big| \, \bX\right]$ is a separable function of $X_i$'s. 
\end{proof}
In the following theorem, we introduce another necessary and sufficient condition under which the lower bound becomes tight. Before stating the result, let us define the convex function  $h(\mathbf{z}): \reals^{mp\times 1} \mapsto \reals$ as
\[
h(\mathbf{z})\triangleq\sum_{i=1}^{p}{\max\lbrace \mathbf{z}_{mi-m+1},\mathbf{z}_{mi-m+2},\ldots,\mathbf{z}_{mi}\rbrace}.
\]
\begin{thm}
\label{thm:SuffNecCond}
Assume $\mathcal{C} \neq \emptyset$. Then, the lower-bound ${\rho_{m}^{\rm lb, \mathcal{C}}}$  is tight
if and only if there exists a solution $\bz^*$ to \eqref{eq:QuadraticProblem} satisfying $h(\bz^*) \leq 1/2$ and $h(-\bz^*)\leq 1/2$. In other words, if and only if the following identity holds:
\begin{align} 
\gamma^{\rm lb} =  \min_{\bz}\;\; & {\bz^T \bQ \bz + \mathbf{f}^T \bz + \frac{1}{4}}, \label{eq:SuffCondThm4}\\
{\rm s.t.} \quad & h(\bz)\leq 1/2, \;{\rm and}\; h(-\bz) \leq 1/2. \nonumber
\end{align}
\end{thm} 

\begin{proof}
The proof can be found in the Appendix.
\end{proof}

%
%
 
Notice that in light of Theorems \ref{thm:SuffNecCond_part1} and \ref{thm:SuffNecCond}, we could identify in polynomial time whether there exists a probability distribution $\bbP\in \mathcal{C}$ with an additive model structure.  Now it is interesting to investigate whether the conditions in  Theorem~\ref{thm:SuffNecCond} are satisfied for most of the classes, or happens for a negligible class of distributions. To  formally study this question, let us start from a probability distribution $\mathbb{P}_0$. Since $\bX \in \mathbb{X}^p$ with $|\mathbb{X}|=m$ and $Y\in \{0,1\}$,  this probability distribution can be uniquely identified by a vector  $\bp_0\in \reals_+^{ 2m^p}$ with $\mathbf{1}^T \bp_0 = 1$, where $\mathbf{1}$ is the vector of all one. Define $\mathcal{C}(\bp_0)$ to be the class of probability distributions having the same pairwise marginals as $\bp_0$. Having this definition in our hands,  the following result follows:
\begin{thm}
For the uniform probability vector  $\tilde{\bp}$, there exists an $\epsilon >0 $ such that for any probability vector $\bp$ with $\|\bp - \tilde{\bp}\|_1<\epsilon$, the class $\mathcal{C}(\bp)$ contains an additive distribution, i.e., $\exists \;\hat{\bbP} \in \mathcal{C}(\bp)$ with $\mathbb{E}_{\hat{\bbP}} [Y \, \big| \, \bX] = \sum_{i=1}^p f_i(X_i)$ for some functions $f_i:\mathbb{X} \mapsto \reals, \;i=1,\ldots,p$.
\end{thm}
\begin{proof}
Let $\tilde{\bp}$ be the uniform distribution over $(\mathbf{X},Y)$, i.e. all $(\mathbf{x},y)$'s take the same probability. Define $\tilde{\bQ} $ and $\tilde{\mathbf{f}}$ to be the matrix and vector in \eqref{eq:QuadraticProblem} obtained for this uniform distribution.  Clearly, $\tilde{\mathbf{f}} = 0$  and therefore the objective of \eqref{eq:QuadraticProblem} is minimized at $\mathbf{z}=\mathbf{0}$ for which $h(\mathbf{z})=h(-\mathbf{z})=0<1/2$.

Note that for any probability distribution and for any $i=1,\ldots ,p$ the columns $m(i-1)+1,\ldots ,mi$ of the defined matrix $\mathbf{Q}$ sum to the unit vector. Therefore,  ${\rm rank}(\tilde{\mathbf{Q}})<(m-1)p+1$. Furthermore, it is easy to check that for the vector $\tilde{\bp}$, ${\rm rank}(\tilde{\mathbf{Q}})$ achieves this upper bound. The reason is that if we subtract the column $mi+j$ from column $mi+j-1$ for any $j=2,\ldots ,m$ and $i=0,\ldots ,p-1$, we obtain a vector taking $\frac{1}{m}$ at $mi+j-1$, $\frac{-1}{m}$ at $mi+j$, and zero elsewhere; which leads to $(m-1)p$ independent vectors. 
Also note that any linear combination of these vectors sums to zero which shows that the dimension of the column space of $\tilde{\mathbf{Q}}$ is at least $(m-1)p+1$, i.e., ${\rm rank}(\tilde{\mathbf{Q}})=(m-1)p+1$. 

Since the function ${\rm rank}(\cdot)$ is lower-continuous, there exists an $\epsilon>0$ for which  ${\rm rank}(\tilde{\mathbf{Q}})\le {\rm rank}({\mathbf{Q}})$ for any $\bQ$ coming from the probability vector $\mathbf{p}$ with $\|\mathbf{p} - \tilde{\mathbf{p}}\|_1 < \epsilon$. Combining the fact that ${\rm rank}(\bQ)\le (m-1)p+1$ for any $\bp$ with the lower continuity of the ${\rm rank}(\cdot)$ function implies that ${\rm rank}(\mathbf{Q})=(m-1)p+1$ for small enough $\epsilon$. Therefore, Moore-Penrose pseudoinverse ${\tilde{\mathbf{Q}}}^{\dagger}$ behaves continuously under small perturbations in probability distribution; see \cite{stewart1969continuity}. Thus for any $\mathbf{p}$ in $\epsilon$-neighborhood of $\tilde{\mathbf{p}}$, we have that $ \big| \Vert {\mathbf{Q}}^{\dagger}\mathbf{f} \|_1- \|\tilde{\mathbf{Q}}^{\dagger}\tilde{\mathbf{f}} \Vert_{1} \big|\le 1/2$. Noticing the fact that $\max\{h(\mathbf{z}^*), h(-\mathbf{z}^*)\}\le \| {\mathbf{Q}}^{\dagger}\mathbf{f} \|_1$ completes the proof.
%
\end{proof}
The above result simply states that the set of distributions $\bp$ leading to the pairwise class $\mathcal{C}(\bp)$  containing an additive distribution has a positive Lebesgue measure over the simplex of probability vectors.  Few remarks are in order:
\begin{rmk}
In the continuous case, the distribution with additive structure always exists in our fixed class of distributions since the jointly Gaussian distribution has additive form and there always exists a jointly Gaussian distribution with a given valid first and second order moments.
\end{rmk}

\begin{rmk}
The proposed lower bound ${\rho_{m}}_{\rm lb}$ is not always a solution to \eqref{eq:MinRenyiPairwise}. Consider the binary valued random variables $X_1$, $X_2$, $Y$ coming from the joint distribution $\tilde{P}$ with
\begin{equation} \nonumber
\begin{split}
& \mathbb{\tilde{P}}_{X_1,X_2,Y}(0,0,0)=0,\; \mathbb{\tilde{P}}_{X_1,X_2,Y}(0,0,1)=0.1 \\
& \mathbb{\tilde{P}}_{X_1,X_2,Y}(1,0,0)=0.2,\; \mathbb{\tilde{P}}_{X_1,X_2,Y}(1,0,1)=0.2 \\
& \mathbb{\tilde{P}}_{X_1,X_2,Y}(0,1,0)=0.1,\; \mathbb{\tilde{P}}_{X_1,X_2,Y}(0,1,1)=0.3 \\
& \mathbb{\tilde{P}}_{X_1,X_2,Y}(1,1,0)=0.1,\; \mathbb{\tilde{P}}_{X_1,X_2,Y}(1,1,1)=0.
\end{split}
\normalsize
\end{equation}  
Now consider the class of distributions obtained by the pairwise marginals of $\tilde{\bbP}$. Then for any distribution $\bbP$ in our class, we have 
\begin{equation}
\nonumber
\begin{split}
&\;\mathbb{P}_{X_1,X_2,Y}(0,0,0) + \mathbb{P}_{X_1,X_2,Y}(1,1,1) \\
=& \; \mathbb{P}_{X_1,X_2}(1,1) -  \mathbb{P}_{X_1,Y}(1,0) + \mathbb{P}_{X_2,Y}(0,0) \\
=&\;\mathbb{\tilde{P}}_{X_1,X_2}(1,1) -  \mathbb{\tilde{P}}_{X_1,Y}(1,0) + \mathbb{\tilde{P}}_{X_2,Y}(0,0)   \\
=& \;\mathbb{\tilde{P}}_{X_1,X_2,Y}(0,0,0) + \mathbb{\tilde{P}}_{X_1,X_2,Y}(1,1,1) = 0.
\end{split}
\end{equation}
Since both of $\mathbb{P}(X_1=0,X_2=0,Y=0),\, \mathbb{P}(X_1=1,X_2=1,Y=1) $ are non-negative, they should be both zero. Combining this fact and the fact that $\tilde{\bbP}$ and $\bbP$ have the same set of pairwise marginals implies that $\bbP = \tilde{\bbP}$. In other words, our class of distributions with the given pairwise marginals is a singleton. Furthermore, 
 the conditional  probability  $\tilde{\bbP}(Y \,\big| \,\bX)$ is given by
\begin{equation}
\begin{split}
\tilde{\mathbb{P}}(Y=1\vert X_1=0,X_2=0) &= 1\\
\tilde{\mathbb{P}}(Y=1\vert X_1=1,X_2=0) &= 1/2\\
\tilde{\mathbb{P}}(Y=1\vert X_1=0,X_2=1) &= 3/4\\
\tilde{\mathbb{P}}(Y=1\vert X_1=1,X_2=1) &= 0\\
\end{split}
\end{equation}
Therefore, $\mathbb{E}(Y\vert X_1=1,X_2=1)+\mathbb{E}(Y\vert X_1=0,X_2=0) \neq \mathbb{E}(Y\vert X_1=1,X_2=0) +\mathbb{E}(Y\vert X_1=0,X_2=1)$ which means that our distribution does not have an additive form, i.e. there is no distribution with additive structure in our class. Hence, this example illustrates a scenario where the proposed lower bound is not tight.
\end{rmk}

\bibliographystyle{IEEEbib}
\bibliography{biblio}
%
%
%
\newpage
\section{Appendix}
\subsection{Proof of Theorem 2}
First let us define  $\mathbf{w} \in \reals^{pm\times 1}$ to be the  indicator random vector for $\mathbf{X}$, i.e., 
\begin{equation}
w_{m(i-1)+k}=\mathbb{I}(X_i=k)=\begin{cases} 
\begin{split}
1 &\quad \text{\rm if}\; X_i=k \\
0 &\quad \text{\rm otherwise}.
\end{split}
\end{cases}
\end{equation}
Clearly,  every separable function of $\bX$ would be a linear function on $\mathbf{W}$ and therefore for every $f\in \mathcal{F}$ there exists a vector $\mathbf{a}$ for which
\begin{equation}
f(\bX)=\mathbf{a}^T\mathbf{w}.
\end{equation}
Also notice that, when $Y$ is binary, the variable $\tilde{Y} = \frac{Y - \mathbb{E}[Y]}{\sqrt{{\rm Var}(Y)}}$ is the only feasible function of $Y$ in \eqref{eq:RenyiDef}. Therefore,
\begin{equation}
\label{eq:tempProof2}
\begin{split}
{\rho_{m}^{\rm lb}}(\bX,Y) = \max_{\mathbf{a}}\quad &{\mathbb{E}\left[\mathbf{a}^T\mathbf{w}\tilde{Y}\right]} \\
{\rm s.t. }\quad   &  \mathbb{E}_\bbP [\mathbf{a}^T\mathbf{w}] = 0, \\
  &\mathbb{E}_\bbP [\mathbf{a}^T\mathbf{w}\mathbf{w}^T\mathbf{a}]  = 1. 
  \end{split}
\end{equation}
Based on the definition of  $\mathbf{Q}$ and $\mathbf{d}$, for any $\mathbb{P}\in \mathcal{C}$ we have
\begin{equation}
\mathbf{Q}=\mathbb{E}_{\mathbb{P}}(\mathbf{w}\mathbf{w}^T),\quad \mathbf{d^{'}}=\mathbb{E}_{\mathbb{P}}(\tilde{Y}\mathbf{w}).
\end{equation}
where
\begin{equation}
\mathbf{d^{'}}=\frac{\frac{1}{2}\mathbf{d} + \left(\frac{1}{2} - \mathbb{P}(Y=1)\right) \mathbb{E}(\mathbf{w})}{\sqrt{\mathbb{P}(Y=0) \mathbb{P}(Y=1)}}. 
\end{equation}
First notice that  there exists $\mathbf{u^{'}}$ for which $\mathbf{Q}\mathbf{u^{'}}=\mathbf{d^{'}}$ since otherwise the objective function in \eqref{eq:tempProof2} would not be bounded from above. Furthermore,  $\mathbb{E}(\mathbf{w})$ is in the column space of $\mathbf{Q}$ and therefore there exists $\mathbf{u}$  with $\mathbf{Q}\mathbf{u}=\mathbf{d}$. With this in mind, the optimization problem \eqref{eq:tempProof2} can be rewritten as
\begin{equation}
\begin{split}
{\rho_{m}^{\rm lb}}(\bX,Y) = \max_{\mathbf{a}}\quad & \mathbf{a}^T\mathbf{d^{'}} \\
{\rm s.t. }\quad   &  \mathbf{a}^T\mathbb{E} [\mathbf{w}] = 0, \\ \quad &
 \mathbf{a}^T\mathbf{Q}\mathbf{a} \le 1. 
  \end{split}
\end{equation}
So the Lagrangian would be
\begin{equation}
\mathcal{L}(\mathbf{a},\beta,\lambda)= \mathbf{a}^T\mathbf{d^{'}} + \lambda ( 1-  \mathbf{a}^T\mathbf{Q}\mathbf{a} ) + \beta \mathbf{a}^T\mathbb{E} [\mathbf{w}] 
\end{equation}
and thus the Lagrange-dual function is
\begin{equation}
\begin{split}
h(\beta,\lambda)&=  \sup_{\mathbf{a}}\,\mathcal{L}(\mathbf{a},\beta,\lambda) \\
&= \sup_{\mathbf{a}}\, \mathbf{a}^T\mathbf{d^{'}} + \lambda ( 1-  \mathbf{a}^T\mathbf{Q}\mathbf{a} ) + \beta \mathbf{a}^T\mathbb{E} [\mathbf{w}] \\
&= \lambda + \frac{1}{4\lambda}\left(\mathbf{d^{'}}+\beta \mathbb{E} [\mathbf{w}] \right)^T \mathbf{Q}^{\dagger}\left(\mathbf{d^{'}}+\beta \mathbb{E} [\mathbf{w}] \right),
\end{split}
\end{equation}
where $\mathbf{Q}^{\dagger}$ denotes the Moore-Penrose pseudoinverse of $\mathbf{Q}$. Note that $\mathbf{Q}$ is not a full rank matrix and thus is not invertible. Here we have used that $\mathbf{Q}$ is positive semi-definite and that $\mathbf{u}^{'}$ exists where $\mathbf{Q}\mathbf{u}^{'}=\mathbf{d}^{'}$ to show that the above quadratic function in the supremum is upper-bounded, and therefore the last equality holds. In addition, notice that for any $\mathbf{v}^{'}$ where $\mathbf{Q}\mathbf{v}^{'}=\mathbf{d}^{'}$, including $\mathbf{Q}^{\dagger}\mathbf{d}^{'}$ since $\mathbf{u}$ exists that $\mathbf{Q}\mathbf{u}=\mathbf{d}$, we have
\begin{equation}
\begin{split}
&\mathbf{Q}\mathbf{v}^{'}=\mathbf{d}^{'} \\
\Rightarrow\; & \mathbb{E}\left[{\mathbf{w}\mathbf{w}^T}\right]\mathbf{v}^{'} = \mathbb{E}\left[{\tilde{Y}\mathbf{w}}\right] \\
\Rightarrow\; & \mathbf{1}^T\mathbb{E}\left[{\mathbf{w}\mathbf{w}^T}\right]\mathbf{v}^{'} = \mathbf{1}^T\mathbb{E}\left[{\tilde{Y}\mathbf{w}}\right] \\
\Rightarrow\; &\mathbb{E}\left[{ \mathbf{1}^T\mathbf{w}\mathbf{w}^T}\right]\mathbf{v}^{'} = \mathbb{E}\left[{\tilde{Y}(\mathbf{1}^T\mathbf{w})}\right] \\
\Rightarrow\; &\mathbb{E}\left[{ \mathbf{w}^T}\right]\mathbf{v}^{'} = \mathbb{E}[\, {\tilde{Y}} \,] = 0,
\end{split}
\end{equation} 
where $\mathbf{1}$ is a vector with all $1$ entries. Therefore,
\begin{equation} \label{eq: Qpseudo1}
\mathbb{E}\left[{ \mathbf{w}}\right]^T\mathbf{Q}^{\dagger}\mathbf{d}^{'}=0.
\end{equation}
Similarly, we can show
\begin{equation} \label{eq: Qpseudo2}
\begin{split}
&\mathbb{E}\left[{ \mathbf{w}}\right]^T\mathbf{Q}^{\dagger}\mathbb{E}\left[{ \mathbf{w}}\right] = 1, \\
&\mathbb{E}\left[{ \mathbf{w}}\right]^T\mathbf{Q}^{\dagger}\mathbf{d} = \mathbb{P}(Y=1)-\mathbb{P}(Y=0).
\end{split}
\end{equation}
Let $p_1=\mathbb{P}(Y=1)$ and $p_0=\mathbb{P}(Y=0)$. To solve the dual problem, we have
\begin{align}
{\rho_{m}^{\rm lb}}(\bX,Y) & = \min_{\lambda \ge 0,\, \beta}\, h(\beta,\lambda) \nonumber \\
& =    \min_{\beta}\, \min_{\lambda \ge 0}\, h(\beta,\lambda) \nonumber \\
& \stackrel{(a)}{=}   \min_{\beta} \sqrt{\left(\mathbf{d^{'}}+\beta \mathbb{E} [\mathbf{w}] \right)^T \mathbf{Q}^{\dagger}\left(\mathbf{d^{'}}+\beta \mathbb{E} [\mathbf{w}] \right)} \nonumber \\
& = \sqrt{\min_{\beta}\, \left(\mathbf{d^{'}}+\beta \mathbb{E} [\mathbf{w}] \right)^T \mathbf{Q}^{\dagger}\left(\mathbf{d^{'}}+\beta \mathbb{E} [\mathbf{w}] \right)} \nonumber \\
& \stackrel{(b)}{=}  \sqrt{ \mathbf{d^{'}}^T \mathbf{Q}^{\dagger} \mathbf{d^{'}} +\min_{\beta}\, \beta^{2} } \nonumber \\
& = \sqrt{ \mathbf{d^{'}}^T \mathbf{Q}^{\dagger} \mathbf{d^{'}}} 
 \\
& \stackrel{(c)}{=} \sqrt{\frac{\mathbf{d}^T \mathbf{Q}^{\dagger} \mathbf{d} -2\left(1 - 2p_1\right)^2  + \left(1 - 2p_1\right)^2  }{4p_0p_1} } \nonumber \\
& \stackrel{(d)}{=} \sqrt{\frac{1-4\gamma^{\text{\rm lb}} - \left(1 - 2p_1\right)^2  }{4p_0p_1} } \nonumber \\
& = \sqrt{1 - \frac{\gamma^{\text{\rm lb}}}{p_0p_1}}. \nonumber
\end{align}
Here $(a)$ uses the fact that, for any $\eta > 0$, the minimum value of $x+\frac{\eta}{4x}$ over $x\in \reals^{+}$ is $\sqrt{\eta}$. $(b)$ and $(c)$ follow from \eqref{eq: Qpseudo1}, \eqref{eq: Qpseudo2}, and that $p_0+p_1=1$. $(d)$ holds since, as shown above, $\mathbf{Q}$ is a positive-semidefinite matrix whose column space includes $\mathbf{d}$, hence
\begin{equation}
\gamma^{\text{\rm lb}}=\frac{1}{4}\left(1 - \mathbf{d}^T \mathbf{Q}^{\dagger} \mathbf{d} \right).
\end{equation} 
Therefore, the proof is complete.

\subsection{Proof of Theorem 4}
Similar to the proof provided for Theorem 2, let $\mathbf{w}\in \reals^{pm\times 1}$ be the vector of indicator variables, i.e. $w_{(i-1)m+k}=\mathbb{I}(X_i=k)$ is the indicator variable for $X_i=k$. Define $b=Y-\frac{1}{2}$. It can be seen that
\begin{equation}
0\le \mathbb{E}_{\mathbb{P}}\left[ \left( \mathbf{w}^T\mathbf{z}-b \right)^2\right] = {\bz^T \bQ \bz + \mathbf{d}^T \bz + \frac{1}{4}},
\end{equation}
for any $\mathbb{P} \in \mathcal{C}$. Therefore, assuming $\mathcal{C} \neq \emptyset$, the above quadratic function takes its minimum value at any $\bz$ satisfying 
\begin{equation}
2\mathbf{Q}\bz + \mathbf{d} =\mathbf{0}.
\end{equation}
Due to Theorem \ref{thm:SuffNecCond_part1}, if ${\rho_{m}^{\rm lb,\mathcal{C}}}$ becomes tight  there exists a $\mathbb{P}^*\in \mathcal{C}$ where
\begin{equation}
\mathbb{E}_{\bbP^{*}} \left[Y\, \big| \, \bX\right] = \sum_{i=1}^p f^*_i(X_i).
\end{equation}
Let $\bz^*= (\bz_i^*)$ is defined by $z^{*}_{(i-1)m+k}=\frac{1}{2p}-f^*_i(k)$. Then, there exist $\bx^*$ and $\bx^{**}$ for which
\begin{equation}
\begin{split}
h(\bz^*)&=\frac{1}{2}-\,\mathbb{E}_{\bbP^{*}} \left[Y\, \big| \, \bX=\bx^*\right] \\
& = \frac{1}{2}-\,{\bbP^{*}}(Y=1\vert \bX=\bx^*) \le \frac{1}{2}
\end{split}
\end{equation}
and
\begin{equation}
\begin{split}
h(-\bz^*)&=-\frac{1}{2}+\,\mathbb{E}_{\bbP^{*}} \left[Y\, \big| \, \bX=\bx^{**}\right] \\
& = -\frac{1}{2}+\,{\bbP^{*}}(Y=1\vert \bX=\bx^{**}) \le \frac{1}{2}.
\end{split}
\end{equation}
Note that the $\left((i-1)m+k\right)$-th entry of $2\mathbf{Q}\bz^* + \mathbf{d}$ is
\begin{align}
&2\sum_{j,x}  \left[ \mathbb{P}^*(X_i=k,X_j=x)(-f^*_j(x)+\frac{1}{2p})\right] + \nonumber \\ 
&\quad \mathbb{P}^*(X_i=k,Y=1)-\mathbb{P}^*(X_i=k,Y=0) = \nonumber \\
& \mathbb{P}^*(X_i=k)-2\sum_{\mathbf{x}:\; x_i=k} \left[  \mathbb{P}^*(\bX=\bx)\mathbb{E}_{\bbP^{*}} \left[Y\, \big| \, \bX=\bx\right]\right] \nonumber \\
&+ \mathbb{P}^*(X_i=k,Y=1) - \mathbb{P}^*(X_i=k,Y=0) = \nonumber \\
& \mathbb{P}^*(X_i=k) -2\,\mathbb{P}^*(X_i=k,Y=1)  + \\
& \quad \mathbb{P}^*(X_i=k,Y=1) -  \mathbb{P}^*(X_i=k,Y=0) =\,0. \nonumber
\end{align}

Therefore, $\bz^*$ satisfies the first order optimality condition of \eqref{eq:QuadraticProblem} with $h(\bz^*)\le 1/2$ and $h(-\bz^*)\le 1/2$. Consequently,  \eqref{eq:SuffCondThm4} holds. 

Now consider the converse direction. Assume there exists such a minimizer $\bz^*$, we can consider the joint distribution $\mathbb{P}^*$ defined as
\begin{equation} \label{NewProbDist}
\mathbb{P}^*(\bX=\bx,Y=y)=\left(\frac{1}{2}-(-1)^y{\bz^*}^T\mathbf{w}^*_{\bx}\right)\mathbb{Q}(\bX=\bx)
\end{equation}
where $\mathbf{w}^*_{\bx}$ denotes the vector of indicator variables for $\bx$, and $\mathbb{Q}$ is a probability  distribution in $\mathcal{C}$ which we supposed is not empty. Note that according to \eqref{NewProbDist}, since $h(\bz^*)\le 1/2$ and $h(-\bz^*)\le 1/2$ hold, $\mathbb{P}^*$ is a valid joint distribution for which we can simply verify
\begin{equation}
\begin{split}
\mathbb{P}^*(X_i=k,X_j=l)&=\mathbb{Q}(X_i=k,X_j=l) \\
&=\mathbb{P}(X_i=k,X_j=l)
\end{split}
\end{equation}
for every $i,j=1,\cdots ,p$ and $k,l=1,\cdots ,m$. Also,
\begin{equation}
\begin{split}
\mathbb{P}^*(X_i=k,Y=1)&=\sum_{\bx : x_i=k} {\left(\frac{1}{2}+{\bz^*}^T\mathbf{w}^*_{\bx}\right)\mathbb{Q}(\bX=\bx)} \\
&=\frac{1}{2}\,\mathbb{P}^*(X_i=k) + (\mathbf{Q}\bz^*)_{(i-1)m+k} \\
&=\frac{1}{2}\,\mathbb{P}^*(X_i=k) - \frac{1}{2}\,\mathbf{f}_{(i-1)m+k} \\
&=\mathbb{P}(X_i=k,Y=1),
\end{split}
\end{equation}
and thus $\mathbb{P}^*\in \mathcal{C}$. Notice that
\begin{equation}
\mathbb{P}^*(Y=1\vert \bX=\bx)=\frac{1}{2}+{\bz^*}^T\mathbf{w}^*_{\bx}
\end{equation}
which shows $\mathbb{P}^*$ has a separable conditional expectation, and hence due to Theorem \ref{thm:SuffNecCond_part1}, the lower-bound ${\rho_{m}^{\rm lb,\mathcal{C}}}$ is tight. 

\end{document}